\documentclass[letterpaper, 12pt]{article}[2000/05/19]
\usepackage[english]{babel}
\usepackage{amsfonts,amsmath,amssymb,amsthm,latexsym,amscd,mathrsfs}
\usepackage{ifthen,cite}
\usepackage[bookmarksnumbered=true]{hyperref}
\hypersetup{pdfpagetransition={Split}}

\newcommand{\evenhead}{Author \ name}
\newcommand{\oddhead}{Article \ name}
\newcommand{\theArticleName}{Article \ name}

\newcommand{\FirstPageHeading}[1]{\thispagestyle{empty}%
\noindent\raisebox{0pt}[0pt][0pt]{\makebox[\textwidth]{\protect\footnotesize \sf }}\par}

\newcommand{\ArticleName}[1]{\renewcommand{\theArticleName}{#1}\vspace{-2mm}\par\noindent {\LARGE\bf  #1\par}}
\newcommand{\Author}[1]{\vspace{5mm}\par\noindent {\Large  #1\par} \par\vspace{2mm}\par}
\newcommand{\Address}[1]{\vspace{2mm}\par\noindent {\it #1} \par}
\newcommand{\Email}[1]{\ifthenelse{\equal{#1}{}}{}{\par\noindent {\rm E-mail: }{\it  #1} \par}}
\newcommand{\URLaddress}[1]{\ifthenelse{\equal{#1}{}}{}{\par\noindent {\rm URL: }{\tt  #1} \par}}
\newcommand{\EmailD}[1]{\ifthenelse{\equal{#1}{}}{}{\par\noindent {$\phantom{\dag}$~\rm E-mail: }{\it  #1} \par}}
\newcommand{\URLaddressD}[1]{\ifthenelse{\equal{#1}{}}{}{\par\noindent {$\phantom{\dag}$~\rm URL: }{\tt  #1} \par}}

\newcommand{\Abstract}[1]{\vspace{6mm}\par\noindent\hspace*{10mm}
\parbox{140mm}{\small {\bf Abstract.} #1}\par}
\newcommand{\Keywords}[1]{\vspace{3mm}\par\noindent\hspace*{10mm}
\parbox{140mm}{\small {\bf Key words:} \rm #1}\par}
\newcommand{\Classification}[1]{\vspace{3mm}\par\noindent\hspace*{10mm}
\parbox{140mm}{\small {\it 2000 Mathematics Subject Classification:} \rm #1}\vspace{3mm}\par}
\newcommand{\ShortArticleName}[1]{\renewcommand{\oddhead}{#1}}
\newcommand{\AuthorNameForHeading}[1]{\renewcommand{\evenhead}{#1}}

\long\def\@makecaption#1#2{
  \sbox\@tempboxa{\small \textbf{#1.}\ \ #2}%
  \ifdim \wd\@tempboxa >\hsize
    {\small \textbf{#1.}\ \ #2}\par \else
    \global \@minipagefalse
    \hb@xt@\hsize{\hfil\box\@tempboxa\hfil}%
  \fi \vskip\belowcaptionskip}

\def\numberwithin#1#2{\@ifundefined{c@#1}{\@nocounterr{#1}}{%
  \@ifundefined{c@#2}{\@nocnterr{#2}}{%
  \@addtoreset{#1}{#2}%
  \toks@\@xp\@xp\@xp{\csname the#1\endcsname}%
  \@xp\xdef\csname the#1\endcsname
    {\@xp\@nx\csname the#2\endcsname.\the\toks@}}}}
\def\E^#1{{\buildrel #1 \over\vee}}
\newtheorem{theorem}{Theorem}
\newtheorem{lemma}{Lemma}

{\theoremstyle{definition}

}

\begin{document}

\FirstPageHeading{V.I. Gerasimenko, Zh.A. Tsvir}

\ShortArticleName{Mean field asymptotics}

\AuthorNameForHeading{V.I. Gerasimenko, Zh.A. Tsvir}

\ArticleName{Mean Field Asymptotics of\\ Generalized Quantum Kinetic Equation}

\Author{V.I. Gerasimenko$^\ast$\footnote{E-mail: \emph{gerasym@imath.kiev.ua}}
        and Zh.A. Tsvir$^\ast$$^\ast$\footnote{E-mail: \emph{Zhanna.Tsvir@simcorp.com}}}

\Address{$^\ast$\hspace*{2mm}Institute of Mathematics of NAS of Ukraine,\\
    \hspace*{4mm}3, Tereshchenkivs'ka Str.,\\
    \hspace*{4mm}01601, Kyiv-4, Ukraine}

\Address{$^\ast$$^\ast$Taras Shevchenko National University of Kyiv,\\
    \hspace*{4mm}Department of Mechanics and Mathematics,\\
    \hspace*{4mm}2, Academician Glushkov Av.,\\
    \hspace*{4mm}03187, Kyiv, Ukraine}

\bigskip

\Abstract{We construct the mean field asymptotics of a solution of initial-value
problem of the generalized quantum kinetic equation and a sequence of explicitly
defined functionals of a solution of stated kinetic equation. As a result the
quantum Vlasov kinetic equation is rigorously derived. Moreover, in case of the
presence of correlations of particles at initial time the mean field limit of a
solution of the generalized quantum kinetic equation is constructed.}

\bigskip

\vspace{2pc}
\Classification{35Q40; 35Q82; 47J35; 82C10; 82C40.}

\makeatletter
\renewcommand{\@evenhead}{
\hspace*{-3pt}\raisebox{-7pt}[\headheight][0pt]{\vbox{\hbox to \textwidth {\thepage \hfil \evenhead}\vskip4pt \hrule}}}
\renewcommand{\@oddhead}{
\hspace*{-3pt}\raisebox{-7pt}[\headheight][0pt]{\vbox{\hbox to \textwidth {\oddhead \hfil \thepage}\vskip4pt\hrule}}}
\renewcommand{\@evenfoot}{}
\renewcommand{\@oddfoot}{}
\makeatother

\newpage
\vphantom{math}

\protect\tableofcontents
\vspace{0.7cm}

\section{Introduction}

During the last decade the considerable progress in the rigorous derivation in scaling limits of
quantum kinetic equations, in particular the nonlinear Schr\"{o}dinger equation \cite{AGT}--\cite{Sp80}
and the Gross-Pitaevskii equation \cite{ESchY2}--\cite{S-R} as well as the quantum Boltzmann equation
\cite{BCEP3} is observed.

In paper \cite{GT} we established that, if initial data is completely defined by a one-particle
marginal density operator, then all possible states of infinite-particle systems at arbitrary moment
of time can be described by the generalized quantum kinetic equation within the framework of a
one-particle density operator without any approximations.
The aim of this paper is to construct the mean field (self-consistent field) asymptotics of a solution
of the initial-value problem of the generalized quantum kinetic equation and to extend this result on case
of the kinetic evolution in the presence of initial correlations of quantum particles.

We shortly outline the structure of the paper and the main results.
At first in Section 2 we formulate some definitions and preliminary facts on the description of quantum
kinetic evolution. Then the main results on the mean field scaling limit of a solution of initial-value
problem of the generalized quantum kinetic equation and the marginal functionals of the state are stated.
In Section 3 we prove the main results. We established that the constructed asymptotics of a solution of
the generalized quantum kinetic equation is governed by the quantum Vlasov kinetic equation for the limit
states and the limit marginal functionals of the state are products of a solution of the derived Vlasov
kinetic equation. In Section 4 we consider some consequences and generalizations of the obtained results.
In particular, we extend the quantum kinetic equations on case of the evolution of particle states in the
presence of correlations at initial time. Finally in Section 5 we conclude with some observations and
perspectives for future research.


\section{Mean field dynamics of quantum many-particle systems}
We adduce some definitions and preliminary facts about the description of quantum dynamics within the
framework of a one-particle density operator governed by the generalized quantum kinetic equation. Then
the main results about its mean field scaling limit is formulated in case of quantum particles obeying
the Maxwell-Boltzmann statistics.

\subsection{The generalized quantum kinetic equation}
We consider a quantum system of a non-fixed (i.e. arbitrary but finite) number of identical
(spinless) particles obeying Maxwell-Boltzmann statistics in the space $\mathbb{R}^{\nu}$.
We will use units where $h={2\pi\hbar}=1$ is a Planck constant,  and $m=1$ is the mass of particles.

Let $\mathcal{H}$ be a one-particle Hilbert space, then the $n$-particle spaces $\mathcal{H}_n$,
are tensor products of $n$ Hilbert spaces $\mathcal{H}$. We adopt the usual convention that
$\mathcal{H}^{\otimes 0}=\mathbb{C}$. We denote by
$\mathcal{F}_{\mathcal{H}}={\bigoplus\limits}_{n=0}^{\infty}\mathcal{H}_{n}$ the Fock space over
the Hilbert space $\mathcal{H}$.

The Hamiltonian $H_{n}$ of $n$-particle system is a self-adjoint operator with domain
$\mathcal{D}(H_{n})\subset\mathcal{H}_{n}$:
\begin{eqnarray}\label{H}
    &&H_{n}=\sum\limits_{i=1}^{n}K(i)+\epsilon\sum\limits_{i_{1}<i_{2}=1}^{n}\Phi(i_{1},i_{2}),
\end{eqnarray}
where $K(i)$ is the operator of a kinetic energy of the $i$ particle, $\Phi(i_{1},i_{2})$ is the
operator of a two-body interaction potential and $\epsilon>0$ is a scaling parameter. The operator
$K(i)$ acts on functions $\psi_n$, that belong to the subspace
$L^{2}_{0}(\mathbb{R}^{\nu n})\subset\mathcal{D}(H_n)\subset L^{2}(\mathbb{R}^{\nu n})$ of infinitely
differentiable functions with compact supports according to the formula:
$K(i)\psi_n=-\frac{1}{2}\Delta_{q_i}\psi_n$. Correspondingly we have:
$\Phi(i_{1},i_{2})\psi_{n}=\Phi(q_{i_{1}},q_{i_{2}})\psi_{n}$, and we assume that the function
$\Phi(q_{i_{1}},q_{i_{2}})$ is symmetric with respect to permutations of its arguments,
translation-invariant and bounded function.

Let $\mathfrak{L}^{1}(\mathcal{F}_\mathcal{H})={\bigoplus\limits}_{n=0}^{\infty}\mathfrak{L}^{1}(\mathcal{H}_{n})$
be the space of sequences $f=(f_0,f_{1},\ldots,f_{n},\ldots)$ of trace class operators
$f_{n}\equiv f_{n}(1,\ldots,n)\in\mathfrak{L}^{1}(\mathcal{H}_{n})$ and $f_0 \in \mathbb{C}$,
that satisfy the symmetry condition: $f_{n}(1,\ldots,n)=f_{n}(i_{1},\ldots,i_{n})$
for arbitrary $(i_{1},\ldots,i_{n})\in(1,\ldots,n)$, equipped with the norm
\begin{eqnarray*}
    &&\|f\|_{\mathfrak{L}^{1}(\mathcal{F}_\mathcal{H})}=
          \sum\limits_{n=0}^{\infty} \|f_{n}\|_{\mathfrak{L}^{1}(\mathcal{H}_{n})}=
          \sum\limits_{n=0}^{\infty}\mathrm{Tr}_{1,\ldots,n}|f_{n}(1,\ldots,n)|,
\end{eqnarray*}
where $\mathrm{Tr}_{1,\ldots,n}$ are partial traces over $1,\ldots,n$ particles. We denote by
$\mathfrak{L}^{1}_0(\mathcal{F}_\mathcal{H})={\bigoplus\limits}_{n=0}^{\infty}\mathfrak{L}^{1}_0(\mathcal{H}_{n})$
the everywhere dense set of finite sequences of degenerate operators with infinitely differentiable kernels with
compact supports \cite{Kato}.

On the space $\mathfrak{L}^{1}(\mathcal{F}_\mathcal{H})$ we define the group
$\mathcal{G}(-t)=\oplus^{\infty}_{n=0}\mathcal{G}_{n}(-t)$ of operators of the von Neumann equations
\begin{eqnarray}\label{groupG}
    &&\mathcal{G}_{n}(-t)f_n\doteq e^{-itH_{n}}f_n\,e^{itH_{n}}.
\end{eqnarray}
On the space $\mathfrak{L}^{1}(\mathcal{F}_\mathcal{H})$ the mapping (\ref{groupG}): $t\rightarrow\mathcal{G}(-t)f$
is an isometric strongly continuous group which preserves positivity and self-adjointness of operators.
For $f_n\in\mathfrak{L}^{1}_0(\mathcal{H}_n)$ there exists a limit in the
sense of the norm convergence on space $\mathfrak{L}^{1}(\mathcal{H}_{s})$ by which the infinitesimal generator
of the group of evolution operators (\ref{groupG}) is determined as follows
\begin{eqnarray}\label{infOper}
    &&\lim\limits_{t\rightarrow 0}\frac{1}{t}\big(\mathcal{G}_{n}(-t)f_{n}-f_{n}\big)
       =-i(H_{n}f_{n}-f_{n}H_{n})\doteq-\mathcal{N}_{n}f_{n},
\end{eqnarray}
where $H_{n}$ is the Hamiltonian (\ref{H}) and the operator: $-i(H_{n}f_{n}-f_{n}H_{n})$ is defined on the
domain $\mathcal{D}(H_{n})\subset\mathcal{H}_{n}$. We denote by $(-\mathcal{N}_{\mathrm{int}}(i,j))$ the operator
\begin{eqnarray}\label{cint}
    &&(-\mathcal{N}_{\mathrm{int}}(i,j))f_n\doteq-i\big(\Phi(i,j)\,f_n-f_n\,\Phi(i,j)\big)
\end{eqnarray}
defined on the subspace $\mathfrak{L}^{1}_0(\mathcal{H}_n)$.

Let us denote $Y\equiv(1,\ldots,s)$, $X\setminus Y\equiv(s+1,\ldots,s+n)$ and $\{Y\}$ is the set
consisting of one element $Y=(1,\ldots,s)$, the mapping $\theta$ is the declusterization mapping
defined by the formula: $\theta(\{Y\},X\setminus Y)=X$.
We define the $(1+n)th$-order ($n\geq0$) cumulant of groups of operators (\ref{groupG}) as follows \cite{GerS}
\begin{eqnarray}\label{cumulant}
   &&\hskip-12mm\mathfrak{A}_{1+n}(t,\{Y\},\,X\setminus Y)=
     \sum\limits_{\mathrm{P}\,:(\{Y\},\,X\setminus Y)=
     {\bigcup\limits}_i X_i}(-1)^{|\mathrm{P}|-1}(|\mathrm{P}|-1)!
     \prod_{X_i\subset\mathrm{P}}\mathcal{G}_{|\theta(X_i)|}(-t,\theta(X_i)),
\end{eqnarray}
where ${\sum\limits}_\mathrm{P}$ is the sum over all possible partitions $\mathrm{P}$ of the set
$(\{Y\},X\setminus Y)=(\{Y\},s+1,\ldots,s+n)$ into $|\mathrm{P}|$ nonempty mutually disjoint subsets
$X_i\subset(\{Y\},X\setminus Y)$, for example,
\begin{eqnarray*}
   &&\mathfrak{A}_{1}(t,\{Y\})=\mathcal{G}_{s}(-t),\\
   &&\mathfrak{A}_{2}(t,\{Y\},s+1)=\mathcal{G}_{s+1}(-t,Y,s+1)-\mathcal{G}_{s}(-t,Y)\mathcal{G}_{1}(-t,s+1).
\end{eqnarray*}

We indicate some properties of operators (\ref{cumulant}).
If $n=0$, the generator of the first-order cumulant $\mathfrak{A}_{1}(t,\{Y\})=\mathcal{G}_{s}(-t)$
for $f_{s}\in\mathfrak{L}_{0}^{1}(\mathcal{H}_{s})\subset\mathfrak{L}^{1}(\mathcal{H}_{s})$, in the
sense of the norm convergence on space $\mathfrak{L}^{1}(\mathcal{H}_{s})$, is given by the operator
\begin{eqnarray*}
   &&\lim\limits_{t\rightarrow 0}\frac{1}{t}\big(\mathfrak{A}_{1}(t,\{Y\})-I\big)f_{s}
       =-\mathcal{N}_{s}f_{s}.
\end{eqnarray*}
In the case $n=1$ we have in the sense of the norm convergence in $\mathfrak{L}^{1}(\mathcal{H}_{s+1})$
\begin{eqnarray*}
    &&\lim\limits_{t\rightarrow 0}\frac{1}{t}\,\mathfrak{A}_{2}(t,\{Y\},s+1)f_{s+1}
        =-\epsilon\sum_{i=1}^{s}\mathcal{N}_{\mathrm{int}}(i,s+1)f_{s+1},
\end{eqnarray*}
where the operator $(-\mathcal{N}_{\mathrm{int}}(i,s+1))$ is defined by formula (\ref{cint}),
and for $n>1$ in the sense of the norm convergence in $\mathfrak{L}^{1}(\mathcal{H}_{s+n})$,
as a consequence that we consider a system of particles interacting by a two-body potential,
it holds
\begin{eqnarray*}
    &&\lim\limits_{t\rightarrow 0}\frac{1}{t}\mathfrak{A}_{1+n}(t)f_{s+n}=0.
\end{eqnarray*}

We consider the mean field (self-consistent field) asymptotic behavior of a solution of the Cauchy problem of the
generalized quantum kinetic equation \cite{GT}
\begin{eqnarray}
  \label{gke}
    &&\hskip-7mm \frac{d}{dt}F_{1}(t,1)=-\mathcal{N}_{1}(1)F_{1}(t,1)+\\
    &&\hskip-7mm +\epsilon\,\mathrm{Tr}_{2}\big(-\mathcal{N}_{\mathrm{int}}(1,2)\big)
        \sum\limits_{n=0}^{\infty}\frac{1}{n!}\,\mathrm{Tr}_{3,\ldots,n+2}
        \mathfrak{V}_{1+n}\big(t,\{1,2\},3,\ldots,n+2\big)\prod _{i=1}^{n+2} F_{1}(t,i),
        \nonumber\\ \nonumber\\
  \label{vpgke}
    &&\hskip-7mm F_1(t,1)|_{t=0}= F_1^0(1).
\end{eqnarray}
In kinetic equation (\ref{gke}) the $(n+1)th$-order generated evolution operator $\mathfrak{V}_{1+n}(t),\,n\geq0$,
is defined as follows (in case of $\{Y\}=\{1,2\}$)
\begin{eqnarray}\label{skrr}
    &&\hskip-7mm \mathfrak{V}_{1+n}(t,\{Y\},X\setminus Y)\doteq \sum_{k=0}^{n}\,(-1)^k\,
        \sum_{n_1=1}^{n}\ldots\sum_{n_k=1}^{n-n_1-\ldots-n_{k-1}}
        \frac{n!}{(n-n_1-\ldots-n_k)!}\times\\
    &&\hskip-7mm \times\widehat{\mathfrak{A}}_{1+n-n_1-\ldots-n_k}(t,\{Y\},s+1,\ldots,s+n-n_1-\ldots-n_k)\times\nonumber\\
    &&\hskip-7mm \times\prod_{j=1}^k\sum\limits_{\mbox{\scriptsize $\begin{array}{c}\mathrm{D}_{j}:Z_j=\bigcup_{l_j} X_{l_j},\\
        |\mathrm{D}_{j}|\leq s+n-n_1-\dots-n_j\end{array}$}}\frac{1}{|\mathrm{D}_{j}|!}
        \sum_{i_1\neq\ldots\neq i_{|\mathrm{D}_{j}|}=1}^{s+n-n_1-\ldots-n_j}
        \prod_{X_{l_j}\subset \mathrm{D}_{j}}\,\frac{1}{|X_{l_j}|!}\,\,
        \widehat{\mathfrak{A}}_{1+|X_{l_j}|}(t,i_{l_j},X_{l_j}),\nonumber
\end{eqnarray}
where $\sum_{\mathrm{D}_{j}:Z_j=\bigcup_{l_j} X_{l_j}}$ is the sum over all possible
dissections of the linearly ordered set $Z_j\equiv(s+n-n_1-\ldots-n_j+1,\ldots,s+n-n_1-\ldots-n_{j-1})$
on no more than $s+n-n_1-\ldots-n_j$ linearly ordered subsets, and we denote by $\widehat{\mathfrak{A}}_{1+n}(t)$
the $(1+n)$-order cumulant of the groups of scattering operators
\begin{eqnarray}\label{so}
   &&\widehat{\mathcal{G}}_{n}(t)=\mathcal{G}_{n}(-t,1,\ldots,n)\prod _{i=1}^{n}\mathcal{G}_{1}(t,i),\quad n\geq1.
\end{eqnarray}
For example,
\begin{eqnarray*}
    &&\hskip-7mm \mathfrak{V}_{1}(t,\{Y\})=\widehat{\mathfrak{A}}_{1}(t,\{Y\}),\\
    &&\hskip-7mm \mathfrak{V}_{2}(t,\{Y\},s+1)=\widehat{\mathfrak{A}}_{2}(t,\{Y\},s+1)-
       \widehat{\mathfrak{A}}_{1}(t,\{Y\})\sum_{i=1}^s\widehat{\mathfrak{A}}_{2}(t,i,s+1).
\end{eqnarray*}

The collision integral series in kinetic equation~(\ref{gke}) converges under the condition that:
$\|F_{1}(t)\|_{\mathfrak{L}^{1}(\mathcal{H})}<e^{-8}.$

The global in time solution of initial-value problem (\ref{gke})-(\ref{vpgke}) is determined by the following
expansion \cite{GT}
\begin{eqnarray}\label{ske}
    &&F_{1}(t,1)= \sum\limits_{n=0}^{\infty}\frac{1}{n!}\,\mathrm{Tr}_{2,\ldots,{1+n}}\,\,
        \mathfrak{A}_{1+n}(t,1,\ldots,n+1)\prod _{i=1}^{n+1}F_{1}^0(i),
\end{eqnarray}
where $\mathfrak{A}_{1+n}(t)$ is the $(1+n)th$-order cumulant (\ref{cumulant}) of groups of operators
(\ref{groupG}). The series (\ref{ske}) converges under the condition that:
$\|F_1^0\|_{\mathfrak{L}^{1}(\mathcal{H})}<e^{-10}(1+e^{-9})^{-1}$.

In case of initial data given in means of a one-particle density operator the evolution of all possible
states of quantum many-particle systems is described by a solution of initial-value problem of the
generalized quantum kinetic equation (\ref{gke})-(\ref{vpgke}) and a sequence of explicitly defined
functionals of a solution of this generalized kinetic equation
\begin{eqnarray}\label{f}
    &&\hskip-15mmF_{s}\big(t,Y\mid F_{1}(t)\big)\doteq\sum _{n=0}^{\infty }\frac{1}{n!}\,
       \mathrm{Tr}_{s+1,\ldots,{s+n}}\,
       \mathfrak{V}_{1+n}\big(t,\{Y\},X\setminus Y\big)\prod _{i=1}^{s+n} F_{1}(t,i),\quad s\geq2,
\end{eqnarray}
where the $(n+1)th$-order ($n\geq0$) generated evolution operator $\mathfrak{V}_{1+n}(t)$ is defined by formula (\ref{skrr}).
Marginal functional series~(\ref{f}) converges under the condition that:
$\|F_{1}(t)\|_{\mathfrak{L}^{1}(\mathcal{H})}<e^{-(3s+2)}$.

\subsection{The mean field limit theorems}
The mean field scaling limit of a solution of the initial-value problem of generalized kinetic equation
(\ref{gke}) is described by the following limit theorem.
\begin{theorem}
Let there exists the limit $f_{1}^0\in\mathfrak{L}^{1}(\mathcal{H})$ of initial data~(\ref{vpgke})
\begin{eqnarray*}
    &&\lim\limits_{\epsilon\rightarrow 0}\big\|\epsilon\,F_{1}^0-f_{1}^0\big\|_{\mathfrak{L}^{1}(\mathcal{H})}=0,
\end{eqnarray*}
then for finite time interval $t\in(-t_{0},t_{0}),$ where
$t_{0}\equiv\big(2\,\|\Phi\|_{\mathfrak{L}(\mathcal{H}_{2})}\|f_1^0\|_{\mathfrak{L}^{1}(\mathcal{H})}\big)^{-1},$
there exists the following limit of solution~(\ref{ske}) of the generalized quantum kinetic equation~(\ref{gke})
\begin{eqnarray}\label{1lim}
    &&\lim\limits_{\epsilon\rightarrow 0}\big\|\epsilon\,F_{1}(t)-
       f_{1}(t)\big\|_{\mathfrak{L}^{1}(\mathcal{H})}=0,
\end{eqnarray}
where the limit one-particle marginal operator $f_{1}(t)$ is represented in the form
\begin{eqnarray}\label{viter}
   &&\hskip-7mm f_{1}(t,1)= \sum\limits_{n=0}^{\infty}\int\limits_0^tdt_{1}\ldots\int\limits_0^{t_{n-1}}dt_{n}
      \mathrm{Tr}_{\mathrm{2,\ldots,1+n}}\mathcal{G}_{1}(-t+t_{1},1)\big(-\mathcal{N}_{\mathrm{int}}(1,2)\big)\times\\
   &&\hskip+5mm\times\prod\limits_{j_1=1}^{2}\mathcal{G}_{1}(-t_{1}+t_{2},j_1)\ldots\prod\limits_{j_{n-1}=1}^{n}
      \mathcal{G}_{1}(-t_{n-1}+t_{n},j_{n-1})\times\nonumber\\
   &&\hskip+5mm \times\sum\limits_{i_{n}=1}^{n}\big(-\mathcal{N}_{\mathrm{int}}(i_{n},1+n)\big)
      \prod\limits_{j_n=1}^{1+n}\mathcal{G}_{1}(-t_{n},j_n)\prod\limits_{i=1}^{1+n}f_{1}^0(i)\nonumber.
\end{eqnarray}
\end{theorem}

For bounded interaction potentials (\ref{H}) series (\ref{viter}) is norm convergent on the space
$\mathfrak{L}^{1}(\mathcal{H})$ under the condition
\begin{eqnarray*}
    &&t<t_{0}\equiv\big(2\,\|\Phi\|_{\mathfrak{L}(\mathcal{H}_{2})}
       \|f_1^0\|_{\mathfrak{L}^{1}(\mathcal{H})}\big)^{-1},
\end{eqnarray*}
and for initial data $f_{1}^0\in\mathfrak{L}^{1}_{0}(\mathcal{H}_{1})$ the scaling limit
operator $f_{1}(t)$ is a strong solution of Cauchy problem of the quantum Vlasov equation
\begin{eqnarray}\label{Vlasov1}
     &&\frac{\partial}{\partial t}f_{1}(t,1)=-\mathcal{N}_{1}(1)f_{1}(t,1)+
       \mathrm{Tr}_{2}\big(-\mathcal{N}_{\mathrm{int}}(1,2)\big)f_{1}(t,1)f_{1}(t,2),\\ \nonumber\\
   \label{Vlasov2}
     &&f_{1}(t)|_{t=0}=f_{1}^0.
\end{eqnarray}

Since a solution of initial-value problem (\ref{gke})-(\ref{vpgke}) of the generalized kinetic equation
converges to a solution of initial-value problem (\ref{Vlasov1})-(\ref{Vlasov2}) of the quantum Vlasov
kinetic equation as (\ref{1lim}), for marginal functionals (\ref{f}) we establish
\begin{theorem}
Under the conditions of Theorem 1 for functionals (\ref{f}) it holds
\begin{eqnarray*}
    &&\lim\limits_{\epsilon\rightarrow 0} \big\|\epsilon^{s} F_{s}\big(t,1,\ldots,s \mid F_{1}(t)\big)-
       \prod\limits_{j=1}^{s}f_{1}(t,j)\big\|_{\mathfrak{L}^{1}(\mathcal{H}_{s})}=0,
\end{eqnarray*}
where the operator $f_{1}(t)$ is defined by series (\ref{viter}).
\end{theorem}

This statement means that in the mean field scaling limit a chaos property preserves in time.

In section 4 these theorems are extended on the case of kinetic evolution of quantum states in the
presence of correlations of particles at initial time.


\section{The mean field limit of a solution of the generalized quantum kinetic equation}
We construct the mean-field scaling limit of a solution of initial-value problem of
generalized kinetic equation (\ref{gke}) and marginal functionals of the state (\ref{f})
and prove stated above limit theorems. On the basis of obtained results we consider
the problem of the justification of the nonlinear Schr\"{o}dinger equation.

\subsection{Preliminaries: cumulants of asymptotically perturbed groups of operators}
For asymptotically perturbed first-order cumulant (\ref{cumulant}) the following statement
is true \cite{Kato}.
\begin{lemma}
If $f_{s}\in\mathfrak{L}^{1}(\mathcal{H}_{s})$, then for arbitrary finite
time interval for the strongly continuous group (\ref{groupG}) it holds
\begin{eqnarray*}\label{lemma1}
  &&\lim\limits_{\epsilon\rightarrow 0}\big\|\mathcal{G}_{s}(-t)f_{s}-
     \prod\limits_{j=1}^{s}\mathcal{G}_{1}(-t,j)f_{s}\big\|_{\mathfrak{L}^{1}(\mathcal{H}_{s})}=0,
\end{eqnarray*}
\end{lemma}
\begin{proof}
If an interaction potential is a bounded operator, then for $f_s\in\mathfrak{L}^{1}(\mathcal{H}_s)$
an analog of the Duhamel formula for group (\ref{groupG}) holds
\begin{eqnarray}\label{iter2kum}
    &&\big(\mathcal{G}_{s}(-t,1,\ldots,s)-
         \prod\limits_{l=1}^{s}\mathcal{G}_{1}(-t,l)\big)f_s=\nonumber\\
    &&=\epsilon\int\limits_{0}^{t}d\tau\prod\limits_{l=1}^{s}\mathcal{G}_{1}(-t+\tau,l)
         \big(-\sum\limits_{i<j=1}^{s}\mathcal{N}_{\mathrm{int}}(i,j)\big)\mathcal{G}_{s}(-\tau)f_s.
\end{eqnarray}
Indeed, the Duamel equation \eqref{iter2kum} is valid for
$f_{s}\in\mathfrak{L}_0^{1}(\mathcal{H}_{s})\subset\mathfrak{L}^{1}(\mathcal{H}_{s})$.
Since the operators from both sides of this equality
are bounded and the set $\mathfrak{L}^{1}_{0}(\mathcal{H}_s)$ is everywhere dense set in the space
$\mathfrak{L}^{1}(\mathcal{H}_s)$, equality (\ref{iter2kum}) holds for arbitrary $f_s\in\mathfrak{L}^{1}(\mathcal{H}_s)$.
We note that the integral in \eqref{iter2kum} exists in strong sense and the operator
${\prod_{l=1}^{s}}\mathcal{G}_{1}(-t+\tau,l)
\big(-{\sum_{i<j=1}^{s}}\mathcal{N}_{\mathrm{int}}(i,j)\big)\mathcal{G}_{s}(-\tau)f_{s}$ is strongly continuous
over $\tau$ for every $f_{s}\in\mathfrak{L}^{1}(\mathcal{H}_{s})$, and hence it is integrable.

Therefore the validity of lemma follows from the estimate
\begin{eqnarray*}
    &&\big\|\big(\mathcal{G}_{s}(-t,1,\ldots,s)-
           \prod\limits_{l=1}^{s}\mathcal{G}_{1}(-t,l)\big)f_s\big\|_{\mathfrak{L}^{1}(\mathcal{H}_{s})}\leq\\
    &&\leq\epsilon\int\limits_{0}^{t}d\tau\big\|\prod\limits_{l=1}^{s}\mathcal{G}_{1}(-t+\tau,l)
           \big(-\sum\limits_{i<j=1}^{s}\mathcal{N}_{\mathrm{int}}(i,j)\big)
           \mathcal{G}_{s}(-\tau)f_s\big\|_{\mathfrak{L}^{1}(\mathcal{H}_{s})}\leq\nonumber\\
    &&\leq\epsilon\,t s(s-1)\|\Phi\|_{\mathfrak{L}(\mathcal{H}_{2})}\|f_s\|_{\mathfrak{L}^{1}(\mathcal{H}_{s})}.
\end{eqnarray*}
\end{proof}

In consequence of Lemma 1 for first-order
cumulant (\ref{cumulant}) of scattering operators (\ref{so}) the equality holds
\begin{eqnarray*}\label{conlemma1}
   &&\lim\limits_{\epsilon\rightarrow 0}\big\|\widehat{\mathcal{G}}_{s}(t)f_{s}-f_{s}
       \big\|_{\mathfrak{L}^{1}(\mathcal{H}_{s})}=0.
\end{eqnarray*}

Correspondingly, for cumulants (\ref{cumulant}) of asymptotically perturbed groups of operators we have:
\begin{lemma}
If $f_{s+n}\in\mathfrak{L}^{1}(\mathcal{H}_{s+n})$, then for arbitrary finite time interval
for the $(1+n)th$-order cumulant of strongly continuous groups (\ref{groupG}) it holds
\begin{eqnarray*}\label{lemma2}
    &&\hskip-9mm\lim\limits_{\epsilon\rightarrow 0}\big\|\frac{1}{\epsilon^n}\,
       \frac{1}{n!}\mathfrak{A}_{1+n}(t,\{Y\},X\setminus Y)f_{s+n}-\nonumber\\
    &&\hskip-9mm-\int\limits_0^tdt_{1}\ldots\int\limits_0^{t_{n-1}}dt_{n} \prod\limits_{j=1}^{s}\mathcal{G}_{1}(-t+t_{1},j)
        \sum\limits_{i_{1}=1}^{s}\big(-\mathcal{N}_{\mathrm{int}}(i_{1},s+1)\big)
        \prod\limits_{j_1=1}^{s+1}\mathcal{G}_{1}(-t_{1}+t_{2},j_1)\ldots\nonumber\\
    &&\hskip-9mm\ldots\prod\limits_{j_{n-1}=1}^{s+n-1}\mathcal{G}_{1}(-t_{n-1}+t_{n},j_{n-1})
        \sum\limits_{i_{n}=1}^{s+n-1}\big(-\mathcal{N}_{\mathrm{int}}(i_{n},s+n)\big)
        \prod\limits_{j_n=1}^{s+n}\mathcal{G}_{1}(-t_{n},j_n)f_{s+n}\big\|_{\mathfrak{L}^{1}(\mathcal{H}_{s+n})}=0.
\end{eqnarray*}
\end{lemma}
The validity of this lemma follows from an analog of the Duhamel formula for cumulants
(\ref{cumulant}) of strongly continuous groups (\ref{groupG})
\begin{eqnarray*}\label{Dcum}
    &&\hskip-8mm\frac{1}{n!}\mathrm{Tr}_{s+1,\ldots,s+n}
     \mathfrak{A}_{1+n}(t,\{Y\},X\setminus Y)f_{s+n}=\\
  &&\hskip-8mm =\epsilon^n \int\limits_0^tdt_{1}\ldots\int\limits_0^{t_{n-1}}dt_{n}
     \mathrm{Tr}_{s+1,\ldots,s+n}\mathcal{G}_{s}(-t+t_{1})
     \sum\limits_{i_{1}=1}^{s}\big(-\mathcal{N}_{\mathrm{int}}(i_{1},s+1)\big)\times\nonumber\\
  &&\hskip-8mm \times\mathcal{G}_{s+1}(-t_{1}+t_{2})\ldots\mathcal{G}_{s+n-1}(-t_{n-1}+t_{n})
     \sum\limits_{i_{n}=1}^{s+n-1}\big(-\mathcal{N}_{\mathrm{int}}(i_{n},s+n)\big)\mathcal{G}_{s+n}(-t_{n})f_{s+n}.\nonumber
\end{eqnarray*}
and it is proving similar to previous lemma.

We give instances of analogs of the Duhamel equation for cumulants (\ref{cumulant}) of scattering operators (\ref{so})
and generated evolution operators  (\ref{skrr}).
If an interaction potential is a bounded operator and $f_{s+1}\in\mathfrak{L}^{1}(\mathcal{H}_{s+1})$,
then for the second-order cumulant $\widehat{\mathfrak{A}}_{2}(t,\{Y\},s+1)$ of scattering operators
(\ref{so}) an analog of the Duhamel equation holds
\begin{eqnarray*}\label{dcum}
    &&\hskip-8mm\widehat{\mathfrak{A}}_{2}(t,\{Y\},s+1)f_{s+1}=\int_0^t d\tau\,\mathcal{G}_{s}(-\tau,Y)
       \mathcal{G}_{1}(-\tau,s+1)\sum\limits_{i_1=1}^{s}\big(-\mathcal{N}_{\mathrm{int}}(i_1,s+1)\big)\times\\
    &&\hskip+17mm\times\widehat{\mathcal{G}}_{s+1}(\tau-t,Y,s+1)\prod_{i_2=1}^{s+1}\mathcal{G}_{1}(\tau,i_2)f_{s+1},\nonumber
\end{eqnarray*}
and, consequently, for the second-order generated evolution operator $\mathfrak{V}_{2}(t,\{Y\},s+1)$ we have:
\begin{eqnarray*}\label{Df}
    &&\hskip-8mm\mathfrak{V}_{2}(t,\{Y\},s+1)f_{s+1}\doteq\big(\widehat{\mathfrak{A}}_{2}(t,\{Y\},s+1)-
       \widehat{\mathfrak{A}}_{1}(t,\{Y\})\sum_{i_1=1}^s \widehat{\mathfrak{A}}_{2}(t,i_1,s+1)\big)f_{s+1}=\\
    &&\hskip-8mm=\int_0^t d\tau\,\mathcal{G}_{s}(-\tau,Y)\mathcal{G}_{1}(-\tau,s+1)\big(\sum\limits_{i_1=1}^{s}
       (-\mathcal{N}_{\mathrm{int}}(i_1,s+1))\widehat{\mathcal{G}}_{s+1}(\tau-t,Y,s+1)-\nonumber\\
    &&\hskip-8mm-\widehat{\mathcal{G}}_{s}(\tau-t,Y)\sum\limits_{i_1=1}^{s}(-\mathcal{N}_{\mathrm{int}}(i_1,s+1))
       \widehat{\mathcal{G}}_{2}(\tau-t,i_1,s+1)\big)\prod_{i_2=1}^{s+1}\mathcal{G}_{1}(\tau,i_2)f_{s+1}.\nonumber
\end{eqnarray*}

Then, according to Lemma 2, i.e. formulas of an asymptotic perturbation of cumulants of groups of operators, and
definition (\ref{skrr}) of the generated evolution operators, we establish
\begin{eqnarray*}
    &&\lim\limits_{\epsilon\rightarrow 0}\big\|\big(\mathfrak{V}_{1}(t,\{Y\})-I\big)f_{s}
         \big\|_{\mathfrak{L}^{1}(\mathcal{H}_{s})}=0,
\end{eqnarray*}
and in the general case the following equalities hold
\begin{eqnarray*}
    &&\lim\limits_{\epsilon\rightarrow 0}\big\|\frac{1}{\epsilon^{n}}\,\mathfrak{V}_{1+n}(t,\{Y\},X\setminus Y)f_{s+n}
        \big\|_{\mathfrak{L}^{1}(\mathcal{H}_{s+n})}=0,\quad n\geq1.
\end{eqnarray*}

\subsection{The proof of the limit theorem}
We give a sketch of the prove of Theorem 1.

In view that the series for $\epsilon\,F_1(t)$ converges, in the sense of
the norm convergence on the space $\mathfrak{L}^{1}(\mathcal{H})$, under the condition that:
\begin{eqnarray*}
    &&t<t_{0}\equiv\big(2\,\|\Phi\|_{\mathfrak{L}(\mathcal{H}_{2})}\|
          \epsilon\,F_{1}^0\|_{\mathfrak{L}^{1}(\mathcal{H})}\big)^{-1},
\end{eqnarray*}
then for $t<t_0$ the remainders of solution series (\ref{ske}) and (\ref{viter}) can be made
arbitrary small for sufficient large $n=n_0$ independently of $\epsilon$. Then, according to Lemma 1, Lemma 2
and definition~(\ref{skrr}), for each integer $n$ every term of these series converge term by term.

Let us construct an evolution equation, which satisfies expression (\ref{viter}). We prove that
it is a solution of initial-value problem (\ref{Vlasov1})-(\ref{Vlasov2}) of the quantum
Vlasov kinetic equation.

Taking into account the validity of equality (\ref{infOper}), we differentiate expression
(\ref{viter}) over the time variable in the sense of pointwise convergence of the space $\mathfrak{L}^{1}(\mathcal{H})$
\begin{eqnarray}\label{dife}
   &&\frac{d}{dt}f_{1}(t,1)=-\mathcal{N}(1)f_{1}(t,1)+\\
   &&+\mathrm{Tr}_{2}\big(-\mathcal{N}_{\mathrm{int}}(1,2)\big)
        \sum\limits_{n=0}^{\infty}\int\limits_0^tdt_{1}\ldots\int\limits_0^{t_{n-1}}dt_{n}\,
        \mathrm{Tr}_{3,\ldots,n+2}\prod\limits_{i_1=1}^{2}\mathcal{G}_{1}(-t+t_{1},i_1)\times\nonumber\\
   &&\times\sum\limits_{k_{1}=1}^{2}\big(-\mathcal{N}_{\mathrm{int}}(k_{1},3)\big)
        \prod\limits_{j_1=1}^{3}\mathcal{G}_{1}(-t_{1}+t_{2},j_1)
        \ldots\prod\limits_{i_{n}=1}^{n+1}\mathcal{G}_{1}(-t_{n}+t_{n},i_{n})\times\nonumber\\
   &&\times\sum\limits_{k_{n}=1}^{n+1}\big(-\mathcal{N}_{\mathrm{int}}(k_{n},n+2)\big)
        \prod\limits_{j_n=1}^{n+2}\mathcal{G}_{1}(-t_{n},j_n)\prod\limits_{i=1}^{n+2}f_1^0(i).\nonumber
\end{eqnarray}
Using the product formula for the one-particle marginal density operator $f_{1}(t,i)$ defined
by series (\ref{viter})
\begin{eqnarray*}\label{prodi}
    &&\prod\limits_{i=1}^{k}f_{1}(t,i)=\sum\limits_{n=0}^{\infty}\int\limits_0^t
        dt_{1}\ldots\int\limits_0^{t_{n-1}}dt_{n}\mathrm{Tr}_{k+1,\ldots,k+n}
        \prod\limits_{i_1=1}^{k}\mathcal{G}_{1}(-t+t_{1},i_1)\times\\
    &&\times\sum\limits_{k_{1}=1}^{k}\big(-\mathcal{N}_{\mathrm{int}}(k_{1},k+1)\big)
        \prod\limits_{j_1=1}^{k+1}\mathcal{G}_{1}(-t_{1}+t_{2},j_1)\ldots
        \prod\limits_{i_n=1}^{k+n-1}\mathcal{G}_{1}(-t_{n-1}+t_{n},i_n)\times\nonumber\\
    &&\times\sum\limits_{k_{n}=1}^{k+n-1}\big(-\mathcal{N}_{\mathrm{int}}(k_{n},k+n)\big)
        \prod\limits_{j_n=1}^{k+n}\mathcal{G}_{1}(-t_{n},j_n)\prod \limits_{i=1}^{k+n}f_{1}^0(i),\nonumber
\end{eqnarray*}
where the group property of one-parameter mapping (\ref{groupG}) is applied, we express the second summand
in the right-hand side of equality (\ref{dife}) in terms of operators ${\prod\limits}_{i=1}^{2}f_{1}(t,i)$,
and consequently, we derive kinetic equation (\ref{Vlasov1}).

\subsection{The limit marginal functionals of the state: the propagation of a chaos}
To give a sketch of the prove of Theorem 2 we represent marginal functionals of the state~(\ref{f})
in terms of the marginal correlation functionals $G_{s}\big(t,Y\mid F_{1}(t)\big),\,s\geq2$, namely
\begin{eqnarray*}\label{FG}
   &&F_{s}\big(t,Y\mid F_{1}(t)\big)=
      \sum\limits_{\mbox{\scriptsize$\begin{array}{c}\mathrm{P}:Y=\bigcup_{i}Y_{i}\end{array}$}}
      \prod_{Y_i\subset \mathrm{P}}G_{|Y_i|}\big(t,Y_i\mid F_{1}(t)\big),\quad s\geq2,
\end{eqnarray*}
where ${\sum\limits}_{\mathrm{P}:Y=\bigcup_{i} Y_{i}}$ is the sum over all possible partitions $\mathrm{P}$
of the set $Y\equiv(1,\ldots,s)$ into $|\mathrm{P}|$ nonempty mutually disjoint subsets $Y_i\subset Y$.
The marginal correlation functionals $G_{s}\big(t,Y\mid F_{1}(t)\big),\,s\geq2$, are represented by the
following expansions \cite{GP}
\begin{eqnarray}\label{corf}
    &&\hskip-9mm G_{s}\big(t,Y\mid F_{1}(t)\big)=\sum\limits_{n=0}^{\infty}\frac{1}{n!}\,\mathrm{Tr}_{s+1,\ldots,s+n}
        \mathfrak{V}_{1+n}\big(t,\theta(\{Y\}),X\setminus Y\big)\prod_{i=1}^{s+n}F_{1}(t,i),\quad s\geq2.
\end{eqnarray}
In series (\ref{corf}) it is introduced the notion of the declusterization mapping $\theta: \{Y\}\rightarrow Y$
defined above. Hence in contrast to expansion (\ref{f}) the $n$ term of expansions
(\ref{corf}) of the marginal correlation functional $G_{s}\big(t,Y\mid F_{1}(t)\big)$ is governed by the
$(1+n)th$-order generated evolution operator (\ref{skrr}) of the $(s+n)th$-order cumulants of the scattering
operators. For example, the lower orders generated evolution operators
$\mathfrak{V}_{1+n}\big(t,\theta(\{Y\}),X\setminus Y\big),\,n\geq0$, have the form
\begin{eqnarray*}\label{rrrlsc}
   &&\mathfrak{V}_{1}(t,\theta(\{Y\}))=\widehat{\mathfrak{A}}_{s}(t,\theta(\{Y\}),\\
   &&\mathfrak{V}_{2}(t,\theta(\{Y\}),s+1)=\widehat{\mathfrak{A}}_{s+1}(t,\theta(\{Y\}),s+1)-
       \widehat{\mathfrak{A}}_{s}(t,\theta(\{Y\}))\sum_{i=1}^s \widehat{\mathfrak{A}}_{2}(t,i,s+1),\nonumber
\end{eqnarray*}
and in case of $s=2$, we have
\begin{eqnarray*}
   &&\mathfrak{V}_{1}(t,\theta(\{1,2\}))=\widehat{\mathcal{G}}_{2}(t,1,2)-I.
\end{eqnarray*}

The statement of Theorem 2 is true in consequence of the validity for functionals (\ref{corf})
the equalities
\begin{eqnarray*}
    &&\lim\limits_{\epsilon\rightarrow 0} \big\|\epsilon^{s} G_{s}\big(t,Y \mid F_{1}(t)\big)
      \big\|_{\mathfrak{L}^{1}(\mathcal{H}_{s})}=0, \quad s\geq2.
\end{eqnarray*}
In view of the structure of the generated evolution operators
$\mathfrak{V}_{1+n}\big(t,\theta(\{Y\}),X\setminus Y\big),\,n\geq0,$ of functionals of the state
(\ref{corf}) the last equality is true according to Lemma 1 and Lemma 2.

\subsection{Mean field quantum kinetic equations}
If we consider pure states, i.e. $f_{1}(t)=|\psi_{t}\rangle\langle\psi_{t}|$ is a one-dimensional
projector onto a unit vector $|\psi_{t}\rangle\in\mathcal{H}$ or in terms a kernel of the marginal
operator $f_{1}(t)$: $f_{1}(t,q,q')=\psi(t,q)\psi^{\ast}(t,q')$, then the quantum Vlasov kinetic
equation reduces to the Hartree equation
\begin{eqnarray}\label{Hartree}
    &&i\frac{\partial}{\partial t}\psi(t,q)=-\frac{1}{2}\Delta_{q}\psi(t,q)+
        \int dq'\Phi(q-q')|\psi(t,q')|^{2}\psi(t,q).
\end{eqnarray}

Moreover, for pure states, if it holds
\begin{eqnarray*}
    &&\lim\limits_{\epsilon\rightarrow 0}\big\|\epsilon\,F_{1}^0-
        |\psi_{0}\rangle\langle\psi_{0}|\big\|_{\mathfrak{L}^{1}(\mathcal{H})}=0,
\end{eqnarray*}
the statement of Theorem 2 reads
\begin{eqnarray*}
    &&\lim\limits_{\epsilon\rightarrow 0}\big\|\,\epsilon^{s} F_{s}\big(t \mid F_{1}(t)\big)-
        |\psi_{t}\rangle\langle\psi_{t}|^{\otimes s}\,\big\|_{\mathfrak{L}^{1}(\mathcal{H}_{s})}=0,
\end{eqnarray*}
where $|\psi_{t}\rangle$ is the solution of the nonlinear Hartree equation (\ref{Hartree}) for
initial data $|\psi_{0}\rangle$.

We remark that in case of a system of particles, interacting by the potential which kernel
is the Dirac measure $\Phi(q)=\delta(q)$, the Hartree equation (\ref{Hartree}) is reduced
to the cubic nonlinear Schr\"{o}dinger equation
\begin{eqnarray*}
    &&i\frac{\partial}{\partial t}\psi(t,q)=
       -\frac{1}{2}\Delta_{q}\psi(t,q)+|\psi(t,q)|^{2}\psi(t,q).
\end{eqnarray*}

The obtained results can be generalized on systems of quantum particles
interacting via many-body potentials, i.e. systems with the Hamilton operators
\begin{eqnarray*}\label{H_n}
   &&H_{n}=\sum\limits_{i=1}^{n}K(i)+\sum\limits_{k=2}^{n}\epsilon^{k-1}
     \sum\limits_{i_{1}<\ldots<i_{k}=1}^{n}\Phi^{(k)}(i_{1},\ldots,i_{k}),
\end{eqnarray*}
where the operator $\Phi^{(k)}$ is an operator of a $k$-body interaction potential.

In this case the generalized quantum kinetic equation (\ref{gke})
has the form \cite{GT}
\begin{eqnarray}\label{gkeN}
  &&\hskip-9mm\frac{d}{dt}F_{1}(t,1)=-\mathcal{N}_{1}(1)F_{1}(t,1)+
     \sum\limits_{n=1}^{\infty}\sum_{k=1}^{n}\epsilon^{k}\frac{1}{(n-k)!}\frac{1}{k!}
     \,\mathrm{Tr}_{2,\ldots,n+1}(-\mathcal{N}_{\mathrm{int}}^{(k+1)})(1,\ldots,\\
  &&\hskip+8mm k+1)\mathfrak{V}_{1+n-k}(t,\{1,\ldots,k+1\},k+2,\ldots,n+1)
     \prod _{i=1}^{n+1}F_{1}(t,i),\nonumber
\end{eqnarray}
where the operator $\mathfrak{V}_{1+n-k}(t)$ is the $(1+n-k)th$-order generated
evolution operator (\ref{skrr}) and
\begin{eqnarray*}
    &&(-\mathcal{N}_{\mathrm{int}}^{(k)})(1,\ldots,k)f_k\doteq-i\big(\Phi^{(k)}\,f_k-f_k\,\Phi^{(k)}\big).
\end{eqnarray*}
In case of a $k$-body interaction potential the collision integral of the generalized quantum
kinetic equation (\ref{gkeN}) is given by the norm convergent series under the condition that:
$\|F_{1}(t)\|_{\mathfrak{L}^{1}(\mathcal{H})}<e^{-8}$.

The mean field scaling limit $f_{1}(t)$ of solution~(\ref{ske}) of initial-value problem of
the generalized kinetic equation~(\ref{gkeN}) is described by the limit theorem similar
to Theorem 2 and for initial data $f_{1}^0\in\mathfrak{L}^{1}_{0}(\mathcal{H})$
it is a strong solution of the Cauchy problem of the following Vlasov quantum kinetic equation
\begin{eqnarray*}
 \label{Vlasov1n}
   &&\frac{d}{dt}f_{1}(t,1)=-\mathcal{N}(1)f_{1}(t,1)+\sum\limits_{n=1}^{\infty}\frac{1}{n!}
      \,\mathrm{Tr}_{2,\ldots,n+1}(-\mathcal{N}_{\mathrm{int}}^{(n+1)})(1,\ldots,n+1)
      \prod _{i=1}^{n+1}f_{1}(t,i),\nonumber\\ \nonumber\\
 \label{Vlasovin}
   &&f_1(t)|_{t=0}= f_1^0.
\end{eqnarray*}
Then for a many-body interaction potential the Hartree equation takes the form
\begin{eqnarray*}\label{Heq}
    &&i\frac{\partial}{\partial t} \psi(t,q_1)=
        -\frac{1}{2}\Delta_{q_1}\psi(t,q_1)+ \\
    &&+\sum\limits_{n=1}^{\infty}\frac{1}{n!}\int d q_2\ldots d q_{n+1}
        \Phi^{(n+1)}(q_1,\ldots,q_{n+1})\prod\limits_{i=2}^{n+1}|\psi(t,q_i)|^{2}\psi(t,q_1),
\end{eqnarray*}
and correspondingly, we can derive the nonlinear Schr\"{o}dinger equation with the $2n-1$ power
nonlinear term.


\section{The kinetic evolution involving initial correlations}
One of the advantages of the developed approach is the possibility to construct the kinetic
equations in scaling limits in the presence of correlations of particle states at initial time.

We extend obtained results on case of quantum systems of particles which initial data
specified by initial correlations, for instance, correlations characterizing the condensate states
of particles.

\subsection{Quantum kinetic equations in the presence of initial correlations}
We will consider initial state which is given by the following sequence of marginal density operators:
\begin{eqnarray*}\label{id}
     &&F(0)=\big(1,F_1^0(1),g_{2}(1,2)
        \prod_{i=1}^{2}F_1^0(i),\ldots,g_{n}(1,\ldots,n)\prod_{i=1}^{n}F_1^0(i),\ldots\big),
\end{eqnarray*}
where the bounded operators $g_{n}\in\mathfrak{L}(\mathcal{H}_n),\,n\geq2$, are specified initial
correlations \cite{GP}. Such initial data is typical for the condensed states of quantum gases, for
example, the equilibrium state of the Bose condensate satisfies the weakening of correlation
condition with the correlations which characterize the condensed state \cite{BogLect}.

In this case the one-particle density operator $F_{1}(t)$ is governed by the following generalized
quantum kinetic equation \cite{GT11}
\begin{eqnarray}\label{gkec}
   &&\hskip-8mm\frac{d}{dt}F_{1}(t,1)=-\mathcal{N}(1)F_{1}(t,1)+\\
   &&\hskip+9mm +\epsilon\,\mathrm{Tr}_{2}(-\mathcal{N}_{\mathrm{int}}(1,2))
      \sum\limits_{n=0}^{\infty}\frac{1}{n!}\,\mathrm{Tr}_{3,\ldots,n+2}\,
      \mathfrak{G}_{1+n}(t,\{1,2\},3,\ldots,n+2)\prod_{i=1}^{n+2}F_{1}(t,i),\nonumber
\\ \nonumber\\
  \label{gkece}
    &&\hskip-7mm F_1(t,1)|_{t=0}= F_1^0(1),
\end{eqnarray}
where the $(1+n)th$-order generated evolution operator $\mathfrak{G}_{1+n}(t),\,n\geq0$, is defined
by the following expansion:
\begin{eqnarray}\label{skrrc}
   &&\hskip-7mm\mathfrak{G}_{1+n}(t,\{Y\},X\setminus Y)\doteq n!\,\sum_{k=0}^{n}\,(-1)^k\,\sum_{n_1=1}^{n}\ldots
       \sum_{n_k=1}^{n-n_1-\ldots-n_{k-1}}\frac{1}{(n-n_1-\ldots-n_k)!}\times\\
   &&\hskip-7mm\times\breve{\mathfrak{A}}_{1+n-n_1-\ldots-n_k}(t,\{Y\},s+1,\ldots,
       s+n-n_1-\ldots-n_k)\times\nonumber\\
   &&\hskip-7mm\times\prod_{j=1}^k\,\sum\limits_{\mbox{\scriptsize$\begin{array}{c}
       \mathrm{D}_{j}:Z_j=\bigcup_{l_j}X_{l_j},\\
       |\mathrm{D}_{j}|\leq s+n-n_1-\dots-n_j\end{array}$}}\frac{1}{|\mathrm{D}_{j}|!}
       \sum_{i_1\neq\ldots\neq i_{|\mathrm{D}_{j}|}=1}^{s+n-n_1-\ldots-n_j}\,\,
       \prod_{X_{l_j}\subset \mathrm{D}_{j}}\,\frac{1}{|X_{l_j}|!}\,\,
       \breve{\mathfrak{A}}_{1+|X_{l_j}|}(t,i_{l_j},X_{l_j}).\nonumber
\end{eqnarray}
In formula \eqref{skrrc} we denote by $\sum_{\mathrm{D}_{j}:Z_j=\bigcup_{l_j} X_{l_j}}$ the sum over all possible
dissections of the linearly ordered set $Z_j\equiv(s+n-n_1-\ldots-n_j+1,\ldots,s+n-n_1-\ldots-n_{j-1})$ on
no more than $s+n-n_1-\ldots-n_j$ linearly ordered subsets and we introduce the $(1+n)th$-order scattering cumulants
\begin{eqnarray*}
   &&\breve{\mathfrak{A}}_{1+n}(t,\{Y\},X\setminus Y)\doteq
       \mathfrak{A}_{1+n}(-t,\{Y\},X\setminus Y)g_{1+n}(\{Y\},X\setminus Y)
       \prod_{i=1}^{s+n}\mathfrak{A}_{1}(t,i).
\end{eqnarray*}
For example,
\begin{eqnarray*}
   &&\hskip-8mm\mathfrak{G}_{1}(t,\{Y\})=\breve{\mathfrak{A}}_{1}(t,\{Y\})\doteq\\
   &&\hskip+9mm\doteq\mathfrak{A}_{1}(-t,\{Y\})g_{1}(\{Y\})\prod_{i=1}^{s}\mathfrak{A}_{1}(t,i),
\end{eqnarray*}
and
\begin{eqnarray*}
   &&\hskip-7mm\mathfrak{G}_{2}(t,\{Y\},s+1)=\\
   &&\hskip-7mm =\mathfrak{A}_{2}(-t,\{Y\},s+1)g_{2}(\{Y\},s+1)\prod_{i=1}^{s+1}\mathfrak{A}_{1}(t,i)-\\
   &&-\mathfrak{A}_{1}(-t,\{Y\})g_{1}(\{Y\})\prod_{i=1}^{s}\mathfrak{A}_{1}(t,i)
       \sum_{i=1}^s\mathfrak{A}_{2}(-t,i,s+1)g_{2}(i,s+1)\mathfrak{A}_{1}(t,i)\mathfrak{A}_{1}(t,s+1),
\end{eqnarray*}
where it is used notations accepted above.

The global in time solution of initial-value problem (\ref{gkec})-(\ref{gkece}) is determined by the
series \cite{GT11}
\begin{eqnarray}\label{skec}
    &&\hskip-9mm F_{1}(t,1)=\sum\limits_{n=0}^{\infty}\frac{1}{n!}\,\mathrm{Tr}_{2,\ldots,{1+n}}\,\,
        \mathfrak{A}_{1+n}(t,1,\ldots,n+1)g_{1+n}(1,...,n+1)\prod _{i=1}^{n+1}F_{1}^0(i),
\end{eqnarray}
where $\mathfrak{A}_{1+n}(t)$ is the $(1+n)th$-order cumulant (\ref{cumulant}) of groups of operators
(\ref{groupG}) and the operators $g_{1+n},\, n\geq0,$ are specified initial correlations.
The series (\ref{skec}) converges under the condition that:
$\|F_1^0\|_{\mathfrak{L}^{1}(\mathcal{H})}<e^{-10}(1+e^{-9})^{-1}$.

Correspondingly, the marginal functionals of the state are represented by the following expansions:
\begin{eqnarray}\label{cf}
   &&\hskip-9mm F_{s}(t,Y\mid F_{1}(t))\doteq\sum _{n=0}^{\infty}\frac{1}{n!}\,
      \mathrm{Tr}_{s+1,\ldots,{s+n}}\,\mathfrak{G}_{1+n}(t,\{Y\},X\setminus Y)\prod_{i=1}^{s+n}F_{1}(t,i),
      \quad s\geq2,
\end{eqnarray}
where generated evolution operators of these functionals are defined by formula \eqref{skrrc}.

Thus, the coefficients of generalized quantum kinetic equation \eqref{gkec} and generated evolution
operators \eqref{skrrc} of marginal functionals of the state \eqref{cf} are determined by the operators specified
initial correlations.

\subsection{The mean field evolution of initial correlations}
In case of initial state involving correlations for generated evolution
operator (\ref{skrrc}) of asymptotically perturbed groups of operators
in the mean field limit the following equality is valid
\begin{eqnarray}\label{limc}
  &&\hskip-5mm\lim\limits_{\epsilon\rightarrow 0}\big\|\frac{1}{\epsilon^n}\mathfrak{G}_{1+n}(t,\{Y\},X\setminus Y)
     f_{s+n}\big\|_{\mathfrak{L}^{1}(\mathcal{H}_{s+n})}=0, \quad n\geq1,
\end{eqnarray}
and in case of first-order generated evolution operator (\ref{skrrc}) we have, respectively
\begin{eqnarray}\label{limc1}
  &&\hskip-5mm\lim\limits_{\epsilon\rightarrow 0}\big\|\big(\mathfrak{G}_{1}(t,\{Y\})-
     \prod_{i_1=1}^{s}\mathcal{G}_{1}(-t,i_1)g_{1}(\{Y\})\prod_{i_2=1}^{s}\mathcal{G}_{1}(t,i_2)\big)
     f_{s}\big\|_{\mathfrak{L}^{1}(\mathcal{H}_{s})}=0.
\end{eqnarray}

In view that under the condition that:
$t<t_{0}\equiv\big(2\,\|\Phi\|_{\mathfrak{L}(\mathcal{H}_{2})}\|\epsilon\,F_{1}^0\|_{\mathfrak{L}^{1}(\mathcal{H})}\big)^{-1}$,
the series for $\epsilon\,F_1(t)$ is norm convergent, then for $t<t_0$ the remainders of solution series~(\ref{skec})
can be made arbitrary small for sufficient large $n=n_0$ independently of $\epsilon$. Then, using stated above asymptotic
perturbation formulas, for each integer $n$ every term of this series converges term by term to the limit operator
$f_{1}(t)$ which is represented by the following series
\begin{eqnarray}\label{viterc}
    &&\hskip-9mmf_{1}(t,1)=\\
    &&\hskip-9mm=\sum\limits_{n=0}^{\infty}\int\limits_0^tdt_{1}\ldots\int\limits_0^{t_{n-1}}dt_{n}
        \mathrm{Tr}_{\mathrm{2,\ldots,1+n}}\,\mathcal{G}_{1}(-t+t_{1},1)
        \big(-\mathcal{N}_{\mathrm{int}}(1,2)\big)\prod\limits_{j_1=1}^{2}\mathcal{G}_{1}(-t_{1}+t_{2},j_1)\ldots\nonumber\\
    &&\hskip-9mm\ldots\prod\limits_{j_{n-1}=1}^{n}\mathcal{G}_{1}(-t_{n-1}+t_{n},j_{n-1})
        \sum\limits_{i_{n}=1}^{n}\big(-\mathcal{N}_{\mathrm{int}}(i_{n},1+n)\big)
        \prod\limits_{j_n=1}^{1+n}\mathcal{G}_{1}(-t_{n},j_n)\times\nonumber\\
    &&\hskip-9mm  \times g_{1+n}(1,\ldots,n+1)\prod\limits_{i=1}^{1+n}f_{1}^0(i)\nonumber.
\end{eqnarray}
For bounded interaction potentials series (\ref{viterc}) is norm convergent on the space
$\mathfrak{L}^{1}(\mathcal{H})$ under the condition: $t<t_{0}\equiv\big(2\,\|\Phi\|_{\mathfrak{L}(\mathcal{H}_{2})}
\|f_1^0\|_{\mathfrak{L}^{1}(\mathcal{H})}\big)^{-1}$.

Thus, if there exists the limit $f_{1}^0\in\mathfrak{L}^{1}(\mathcal{H})$ of initial data (\ref{vpgke}), namely
\begin{eqnarray*}
   &&\lim\limits_{\epsilon\rightarrow 0}\big\|\epsilon\,F_{1}^0-f_{1}^0\big\|_{\mathfrak{L}^{1}(\mathcal{H})}=0,
\end{eqnarray*}
then for finite time interval $t\in(-t_{0},t_{0}),$ where
$t_{0}\equiv\big(2\,\|\Phi\|_{\mathfrak{L}(\mathcal{H}_{2})}\|f_1^0\|_{\mathfrak{L}^{1}(\mathcal{H})}\big)^{-1},$
there exists the mean field limit of solution expansion (\ref{skec}) of the generalized quantum kinetic equation (\ref{gkec}):
\begin{eqnarray}\label{1limc}
    &&\lim\limits_{\epsilon\rightarrow 0}\big\|\epsilon\,F_{1}(t)-
       f_{1}(t)\big\|_{\mathfrak{L}^{1}(\mathcal{H})}=0,
\end{eqnarray}
where the operator $f_{1}(t)$ is represented by series (\ref{viterc}) and it is a solution of the Cauchy problem of
the modified Vlasov quantum kinetic equation
\begin{eqnarray}\label{mVe}
  &&\hskip-5mm\frac{d}{dt}f_{1}(t,1)=-\mathcal{N}(1)f_{1}(t,1)+\\
  &&\hskip-5mm+\mathrm{Tr}_{2}(-\mathcal{N}_{\mathrm{int}})(1,2)
     \prod_{i_1=1}^{2}\mathcal{G}_{1}(-t,i_1)g_{1}(\{1,2\})
     \prod_{i_2=1}^{2}\mathcal{G}_{1}(t,i_2)f_{1}(t,1)f_{1}(t,2),\nonumber\\
\label{Vlasov2c}
  &&\hskip-5mmf_{1}(t)|_{t=0}=f_{1}^0.
\end{eqnarray}

Since a solution of initial-value problem (\ref{gkec})-(\ref{gkece}) of the generalized kinetic equation
converges to a solution of initial-value problem (\ref{mVe})-(\ref{Vlasov2c}) of the modified quantum Vlasov
kinetic equation as (\ref{1limc}) and equalities (\ref{limc}) and (\ref{limc1}) hold, for marginal functionals
of the state (\ref{cf}) we establish
\begin{eqnarray*}\label{lf}
  &&\hskip-8mm\lim\limits_{\epsilon\rightarrow 0}\big\|\epsilon^{s}F_{s}(t,Y\mid F_{1}(t))-
     \prod _{i_1=1}^{s}\mathcal{G}_{1}(-t,i_1)g_{1}(\{Y\})
     \prod_{i_2=1}^{s}\mathcal{G}_{1}(t,i_2)
     \prod\limits_{j=1}^{s}f_{1}(t,j)\big\|_{\mathfrak{L}^{1}(\mathcal{H}_{s})}=0.
\end{eqnarray*}
This equality means the propagation of initial correlations in time in the mean field limit.

Let us consider the pure states, i.e. the operator $f_{1}(t)=|\psi_{t}\rangle\langle\psi_{t}|$
is a one-dimensional projector onto a unit vector $|\psi_{t}\rangle\in\mathcal{H}$ and its kernel has
the following form: $f_{1}(t,q,q')=\psi(t,q)\psi^{\ast}(t,q')$. Then, we remark that in case of 
a system of particles, interacting by the potential which kernel $\Phi(q)=\delta(q)$ is the Dirac measure,
the modified quantum Vlasov kinetic equation (\ref{mVe}) reduces to the Gross-Pitaevskii-type kinetic
equation
\begin{eqnarray}\label{gpeqc}
  &&\hskip-12mm i\frac{\partial}{\partial t}\psi(t,q)=-\frac{1}{2}\Delta_{q}\psi(t,q)+
     \int d q'd q''\mathfrak{b}(t,q,q;q',q'')\psi(t,q'')\psi^{\ast}(t,q)\psi(t,q),
\end{eqnarray}
where the coupling ratio $\mathfrak{b}(t,q,q;q',q'')$ of the collision integral
is the kernel of the scattering length operator
$\prod_{i_1=1}^{2}\mathcal{G}_{1}(-t,i_1)b_{1}(\{1,2\})\prod_{i_2=1}^{2}\mathcal{G}_{1}(t,i_2)$.

Observing that on the macroscopic scale of the variation of variables, groups of operators
(\ref{groupG}) of finitely many particles depend on microscopic time variable $\varepsilon^{-1}t$,
where $\varepsilon\geq0$ is a scale parameter, the dimensionless marginal functionals of the state
are represented in the form: $F_{s}(\varepsilon^{-1}t\mid F_{1}(t))$. As a result of the formal
limit processing $\varepsilon\rightarrow0$ in collision integral (\ref{gpeqc}), we establish
the Markovian kinetic evolution with the corresponding coefficient $\mathfrak{b}(\varepsilon^{-1}t)$.

\section{Conclusion}
In the mean field scaling limit we derived the quantum Vlasov kinetic equation and correspondingly,
the Hartree equation (or the nonlinear Schr\"{o}dinger equation) for pure states of quantum
systems of particles obeying the Maxwell-Boltzmann statistics. In particular,
in case of two-body interaction potentials it is the evolution equation with the cubic nonlinear term
and in case of $n$-body interaction potentials the Hartree equation contains the $2n-1$ power nonlinear term.
The obtained results can be extended to quantum systems of bosons and fermions.

The mean field scaling asymptotics of a solution of the generalized quantum kinetic equation
of many-particle systems in condensed states has been also constructed.
We note that one more approach to the construction of the kinetic equations in the mean field limit,
in the case of the presence of correlations at initial time, can be developed on basis of the description
of the kinetic evolution in terms of marginal observables \cite{G11}.


\addcontentsline{toc}{section}{References}
\renewcommand{\refname}{References}

\end{document}